\documentclass[11pt]{article}
\usepackage{geometry}
\usepackage{graphicx}
\usepackage{epsfig}
\usepackage{amsmath}
\usepackage{multirow}
\usepackage{amssymb}
\usepackage{color}
\usepackage{tikz}
\usepackage{caption}
\usepackage{subfig}
\usepackage[affil-it]{authblk}
\usepackage[english]{babel}
\usepackage{blindtext}
\usepackage{latexsym}
\usepackage{algorithm}
\usepackage{algorithmicx} 
\usepackage{algpseudocode}
\newtheorem{proposition}{Proposition}
\newtheorem{definition}{Definition}
\newtheorem{theorem}{Theorem}
\newtheorem{lemma}{Lemma}
\newtheorem{example}{Example}
\newenvironment{proof}{\noindent{\bf Proof:}}{\hfill$\Box$\vspace*{1mm}}
\makeatletter
\def\algbackskip{\hskip-\ALG@thistlm}
\makeatother
\allowdisplaybreaks[4]

\providecommand{\keywords}[1]
{
  \small	
  \textbf{\textit{Keywords---}} #1
}

\title{Discrimination of attractors with noisy nodes in Boolean networks}
\author[1]{Xiaoqing Cheng}
\author[2]{Wai-Ki Ching}
\author[2]{Sini Guo}
\author[3,*]{Tatsuya Akutsu}
\affil[1]{School of Mathematics and Statistics, Xi'an Jiaotong University, Xi'an, 710049, China}
\affil[2]{Advanced Modelling and Applied Computing Laboratory, Department of Mathematics, The University of Hong Kong, Pokfulam Road, Hong Kong}
\affil[2]{Bioinformatics Center, Institute for Chemical Research, Kyoto University, Gokasho, Uji, Kyoto, 611-0011, Japan}
\affil[*]{Corresponding author: Tatsuya Akutsu, takutsu@kuicr.kyoto-u.ac.jp}

\begin{document}
\maketitle

\abstract{Observing the internal state of the whole system using a small number of sensor nodes is important in analysis of complex networks.
Here, we study the problem of determining the minimum number
of sensor nodes to discriminate attractors under the assumption
that each attractor has at most $K$ noisy nodes.
We present exact and approximation algorithms for
this minimization problem.
The effectiveness of the algorithms is also demonstrated by
computational experiments using both synthetic data and
realistic biological data.
}


\section{Introduction}

It is important for analyzing complex network systems
to select a small set of nodes (i.e., \emph{sensor nodes})
whose measurements can determine all other state variables.
Relationships between structure of complex networks and
sensor nodes have been analyzed recently,
especially for linear systems \cite{yan2015spectrum,  liu2013observability}. 
However, biological systems contain non-linear components and thus
exhibit switch-like behaviors.
Therefore, it is essential to study the observability of non-linear systems.
Even though most biological phenomena manifest them in a continuous domain, the binary expression shows promising and useful results \cite{Shmulevich2002, Watterson2008}.
The Boolean network (BN) is one of the most studied mathematical models for genetic networks \cite{kauffman1969b, kauffman1993},
in which the state of each gene is represented by 0 (off) or 1 (on).
Observability of BNs has been widely studied \cite{cheng2010analysis, laschov2013observability, li2015controllability}. 
However, it is impossible in most cases to observe all internal states
from a small set of sensor nodes because BN is a highly non-linear network \cite{li2015controllability}.
Therefore, another approach has been proposed: discrimination of attractors,
where an {\it attractor} is a collection of state cycles.
Attractors
are classified into \emph{singleton attractors} and
\emph{periodic attractors}, where the former and latter correspond to
statically steady states and periodically steady states, respectively.
The purpose of discrimination of attractors is to determine the minimum set of sensor nodes required to discriminate all given attractors.
Since attractors are often interpreted as cell types \cite{huang1999},
the discrimination problem corresponds to a problem of selecting
the minimum number of genes that are needed to identify types of cells
(e.g., types of cancers), which is closely related to selection of
\emph{biomarkers} or \emph{marker genes}, a very important topic
in biological and medical sciences \cite{shen2010, bell2010}.

This discrimination problem was proposed in \cite{Qiu2015},
and has been extensively studied \cite{Cheng2017b}.
All the results assume clean input data.
However, gene expression noise is inevitable due to environmental fluctuations and the stochasticity of biochemical reactions such as transcription, chromatin remodeling and post-translational regulation \cite{chalancon2012interplay, love2014,wu2020orthogonal}.
Therefore, proposing a robust discrimination model is essential towards
robust classification of cell types.
To this end, we reformulate the discrimination problem on BNs
by assuming the number of noisy nodes is bounded by a number $K$,
where this number is closely related to the Hamming distance,
a standard distance measure for binary vectors.

In this paper,
we consider 
the discrimination problem for attractors with noisy nodes firstly,
and present an exact algorithm. Another
polynomial-time approximation algorithm is proposed to in order to balance the tradeoff between the size of the target set and the overall time complexity.
Discrimination of singleton attractors with noisy nodes is a special case of our general discrimination problem here.
In this special case, the distance between any pair of attractors equals to Hamming distance between two attractors' states
and thus it might be possible to develop simpler and/or faster algorithms.
Therefore we present an exact algorithm and a polynomial-time approximation algorithm specified for discrimination of singleton attractors with noisy nodes afterwards.
Finally, we perform computational experiments using
synthetic data and realistic biological data.
We remark that in our study, we assume a set of attractors are given
without knowing the internal structure of a BN. 
Although enumerating all the singleton attractors is an NP-hard problem,
there are some algorithms developed to find
all the attractors up to moderate size networks \cite{akutsu1998,veliz2014,zanudo2013}.
Furthermore, we assume that those attractors can be given independent of BN structures since they will be directly obtained from the expression data of stable cells.

\section{Discrimination of Attractors with Noisy Nodes}

\begin{table}[t]
\caption{List of Notations}
\begin{center}
\resizebox{0.9\textwidth}{!}{
\begin{tabular}{cl}
\hline
\hline
\multicolumn{2}{c}{Common Notations}\\

\hline
\hline
$m$ & Number of attractors\\
$n$ & Number of genes\\
$M$ & Number of POAs (i.e., $M = \binom{m}{2}$)\\
$x_{T(i_1,i_2,m)}, 1\leq T(i_1, i_2, m)\leq M$& A pair of attractors (POA), $(Att_{i_1}, Att_{i_2})$\\
$U=\{x_l, 1\leq l \leq M\}$ & A set of POAs need to be discriminated\\
\hline
\hline
\multicolumn{2}{c}{Notations in {\bf MinDattNN}}\\
\hline
\hline
${\bf v}$ & A 0-1 vector\\
$V$ & A set of nodes corresponding to genes\\
\multirow{2}{*}{${\bf v}_{\hat{V}}, \hat{V}\subseteq V$} & $|\hat{V}|$-dimensional vector consisting of \\
& elements of ${\bf v}$ that correspond to $\hat{V}$\\
$Att_{i_1}=[{\bf v }(0),\ldots,{\bf v }(p(i_1)-1)]$&\multirow{2}{*}{Two periodic attractors}\\
$Att_{i_2}=[{\bf w }(0),\ldots,{\bf w }(p(i_2)-1)]$ & \\
$p(i)$ & The period of $Att_i$\\
$Ser(Att_{i_1}, \hat{V}, t)=$ & An infinite sequence of $|\hat{V}|$-dimensional \\
$[{\bf v}_{\hat{V}}(t), {\bf v}_{\hat{V}}(t+1), {\bf v}_{\hat{V}}(t+2), \ldots]$ & vectors beginning from time step $t$\\
$Dist(Att_{i_1}, Att_{i_2}, \hat{V})$ &
Distance between $Att_{i_1}$ and $Att_{i_2}$ by observing $\hat{V}$\\
$D[T(i_1, i_2, m), T(j_1, j_2, n)]$ & Distance between $(Att_{i_1}, Att_{i_2})$ by observing $\{v_{j_1}, v_{j_2}\}$\\
$s_{T(j_1,j_2, n)}=\{x_{T(i_1, i_2, m)}| D[T(i_1, i_2, m), T(j_1, j_2, n)]\neq 0\}$ & A set of POAs that can be discriminated by $\{v_{j_1}, v_{j_2}\}$\\
$S=\{s_1, s_2, \ldots, s_{\frac{n(n-1)}{2}}\}$ & A set of all candidate sensor pairs\\
$r_{T(i_1, i_2, m)}$ & A dummy variable 
indicating the distance  \\
$\leq Dist(Att_{i_1}, Att_{i_2}, \hat{V})$&between a POA under the current discriminator\\
\multirow{2}{*}{$G_D(T(i_1, i_2, m), \hat{V})$} & Adjacent graph of $T(i_1, i_2,m)-$th row of \\
& $D$ constrained on nodes in $\hat{V}$\\
$MC(G_D(T(i_1, i_2, m), \hat{V}))$ & A maximum clique of $G_D(T(i_1, i_2, m), \hat{V})$\\
\hline
\hline
\multicolumn{2}{c}{Notations in {\bf MinDSattNN}}\\
\hline
\hline

$Att$ & An $m\times n$ attractor matrix\\
$J=\{j_1, j_2, \ldots, j_k\}$ & A set of column/row indices\\
$A[i,-](resp. A[-,j])$ &  The $i$-th row (resp. $j$-th column) of $A$\\
\multirow{2}{*}{$A[i, J](resp. A[J, j])$} & The sub-matrix of $A[i, -]$ (resp. $A[-, j]$) consisting \\
& of the $j_1, j_2, \ldots, j_k$-th columns (resp. rows)\\
$H(x, y)$ & Hamming distance between vectors $x$ and $y$ \\
$s_j=\{x_{T(i_1, i_2, m)}| Att[i_1, j]\neq Att[i_2, j]\}$ & A set of POAs that can be discriminated by $v_j$\\
$S=\{s_1, s_2, \ldots, s_n\}$ & A set of candidate sensor nodes\\
\multirow{2}{*}{$r_{T(i_1, i_2, m)}=H(Att[i_1, J], Att[i_2, J])$} & A dummy variable indicating the distance between\\
& a POA under the current discriminator \\

\hline
\end{tabular}
}
\end{center}
\label{tab:notations}
\end{table}

A list of notations used in this paper is given in Table~\ref{tab:notations}.
Firstly, we give a mathematical formulation of finding a minimum discriminator
for attractors with noisy nodes ({\bf MinDattNN}).
To this end,
we define the distance between a pair of attractors (POA) by observing a set of nodes $\hat{V}$.
This new definition is needed because
a periodic attractor is a periodically steady time series.
For a set $\hat{V}\subseteq V$ and an $n$-dimensional $0$-$1$ vector
${\bf v}$, ${\bf v}_{\hat{V}}$ denotes the $|\hat{V}|$-dimensional vector consisting of elements of ${\bf v}$ that correspond to $\hat{V}$.
For example, if $n=5, {\bf v}=[1,1,0,1,0]$ and
$\hat{V}=\{v_2, v_3, v_5\}$,
then ${\bf v}_{\hat{V}}=[1, 0, 0]$.
Let $Att_{i_1}$ and $Att_{i_2}$ be two periodic attractors
and $p(i)$ be the period of $Att_i$:
$
Att_{i_1}=[{\bf v }(0),{\bf v }(1),\ldots,{\bf v }(p(i_1)-1)]
$
and
$
Att_{i_2}=[{\bf w }(0),{\bf w }(1),\ldots,{\bf w }(p(i_2)-1)].
$
Define $Ser(Att_{i_1}, \hat{V}, t)$ as an infinite sequence of $|\hat{V}|$-dimensional vectors beginning from time step $t$:
$
Ser(Att_{i_1}, \hat{V}, t)=[{\bf v}_{\hat{V}}(t), {\bf v}_{\hat{V}}(t+1), {\bf v}_{\hat{V}}(t+2), \ldots].
$
Let $Dist(Att_{i_1}, Att_{i_2}, \hat{V})$ be the distance between $Att_{i_1}$ and $Att_{i_2}$ by observing $\hat{V}$: \\
\begin{tabular}{ll}
&$Dist(Att_{i_1}, Att_{i_2}, \hat{V})=$ \\
&$\displaystyle \min_{t=0}^{p-1} \bigg\{
\displaystyle \sum_{v_i\in \hat{V}}I\{Ser(Att_{i_1}, \{v_i\}, 0), Ser(Att_{i_2}, \{v_i\}, t)\} \bigg\}$\\
\end{tabular}\\
where $p=LCM(p(i_1), p(i_2))$ (LCM means the least common multiple) and\\
\begin{tabular}{ll}
\ & $I\{Ser(Att_{i_1}, \{v_i\}, 0), Ser(Att_{i_2}, \{v_i\}, t)\}$\\
= & $\left\{
\begin{array}{ll}
1, & \mbox{if} \ Ser(Att_{i_1}, \{v_i\}, 0)\neq Ser(Att_{i_2}, \{v_i\}, t), \\
0, & \mbox{otherwise}.
\end{array}
\right.$
\end{tabular}\\
Then $Dist(Att_{i_1}, Att_{i_2}, V)=0$ if and only if these two attractors are identical.
 In the noisy case, it is hypothesized that noisy nodes vary in
different attractors, thus at most $2K$ nodes' value may flip from $1$ to $0$ or vice versa for a pair of attractors.
Therefore, in the worst case,
a POA is discriminated only if there are at least $2K+1$ different nodes observed, that is $Dist(Att_{i_1}, Att_{i_2}, V)>2K$.
\begin{definition}
(Minimum Discriminator for Attractors with Noisy Nodes 
{\bf [MinDattNN]})\\
{\bf Input}: A set of $m$ attractors $\{Att_1, Att_2, \ldots, Att_m\}$ where
$Att_i$ is a $p(i) \times n$ binary matrix ($n$ is number of genes), and an integer $K$ denoting
the maximum number of noisy nodes per attractor.\\
{\bf Output}: A minimum cardinality set $\hat{V}$ of nodes
such that
$Dist(Att_{i_1}, Att_{i_2}$, $\hat{V}) \geq 2K+1$
holds for all $i_1, i_2$ with $1 \leq i_1< i_2 \leq m$.
\end{definition}

\subsection{Exact Algorithm for {\bf MinDattNN}}

Inspired by Lemma 1 in \cite{Cheng2017b}, we consider gene pairs. We first construct a binary $\binom{m}{2} \times \binom{n}{2}$ matrix $D$ by
$D[T(i_1, i_2, m), T(j_1, j_2, n)]
=Dist(Att_{i_1}, Att_{i_2},$ $\{v_{j_1}, v_{j_2}\})$. Here $T(i_1 ,i_2, m)=(i_1-1)m -\frac{i_1(i_1-1)}{2} +i_2-i_1$. 
Then for each POA $(Att_{i1}, Att_{i_2})$ and a set $\hat{V}$, we construct an undirected graph $G_D(T(i_1, i_2, m), \hat{V})=\langle \hat{V}, \hat{E}\rangle$ where $\hat{E}=\{e=(v_{j_1},v_{j_2}) | v_{j_1}, v_{j_2}\in \hat{V}, D[T(i_1, i_2, m), T(j_1, j_2, n)]=0\}$.  
From Lemma \ref{lem:clique} (given below),
$Dist($ $Att_{i_1}, Att_{i_2}, \hat{V})$ can be calculated by computing a maximum clique of
$G_D(T(i_1, i_2, m), \hat{V})$, then we need to decide whether $Dist(Att_{i_1}, Att_{i_2}, \hat{V})$ $\geq 2K+1$
holds for all $i_1, i_2$ with $1 \leq i_1< i_2 \leq m$. 
Note that the exceptional
case that all node in $G_D(T(i_1, i_2, m), \hat{V})$ are isolated needs to be discussed based on  whether these two attractor can be discriminated by observing any nodes in $\hat{V}$ (line 7-9 in Algorithm \ref{al:exp}). 
Example~\ref{example1} is an illustrative example for calculation of $Dist(Att_{i_1}, Att_{i_2}, \hat{V})$.
The resulting algorithm is given in Algorithm \ref{al:exp},
and its time complexity is analyzed in Theorem \ref{thm:exact}.

\begin{lemma}
Suppose that $G_D(T(i_1, i_2, m), \hat{V})$  has at least one edge.
$\gamma=|MC(G_D(T(i_1, i_2, m), \hat{V}))|$ is the number of nodes in the maximum clique of
$G_D(T(i_1, i_2, m), \hat{V})$.  
Then $Dist(Att_{i_1}, Att_{i_2}, \hat{V})= |\hat{V}|-\gamma$.
\label{lem:clique}
\end{lemma}
\begin{proof}
Define\\
$Dif_{(i_1, i_2)}(\{v_j\}, t)$=
$\left\{
\begin{aligned}
0, & \ \mbox{if}\  Ser(Att_{i_1}, \{v_j\}, 0)
=Ser(Att_{i_2}, \{v_j\}, t),\\
1, & \ \mbox{otherwise},
\end{aligned}
\right.$\\
Without considering order of nodes in $\hat{V}$, $Dif_{(i_1, i_2)}(\hat{V}, t)$ has the form of
$
Dif_{(i_1, i_2)}(\hat{V}, t)=\underbrace{00\cdots0}_{\hat{V}_1(t)}\underbrace{11\cdots1}_{\hat{V}_2(t)}.  
$
From definition, we know $Dist(Att_{i_1}, Att_{i_2},$ $\hat{V})=\displaystyle\min_{t}|\hat{V}_2(t)|=|\hat{V}|-\displaystyle\max_{t}|\hat{V}_1(t)|$. On the other hand, if $v_{j_1}, v_{j_2}\in \hat{V}_1(t)$, then $Dist(Att_{i_1}, Att_{i_2}, \{v_{j_1}, v_{j_2}\})=0$, this means all nodes in $\hat{V}_1(t)$ forms a clique. Maximizing $|\hat{V}_1(t)|$ equals to calculating the number of nodes in a maximum clique of $G_D(T(i_1, i_2, m),\hat{V})$, which completes the proof.
\end{proof}

{\centering
\begin{minipage}{0.9\textwidth}
 \begin{algorithm}[H]
    \caption{Exact algorithm for {\bf MinDattNN}}
    \hspace*{\algorithmicindent} \textbf{Input:}
set of attractors, set of nodes $V$, integer $K$ \\
    \hspace*{\algorithmicindent} \textbf{Output:}
set of nodes $\hat{V}$
    \begin{algorithmic}[1]
    \State Calculate matrix $D$
    \For{$k=2K+1$ to $n$}
    \For{$\hat{V}\subset V$ and $|\hat{V}|=k$}
    \State $sig \gets 0$; 
    \For{$1\leq i_1 \leq i_2 \leq m$}
    \State $sig_2 \gets \displaystyle\sum_{v_{j_1},v_{j_2}\in \hat{V}}D[T(i_1,i_2, m),T(j_1, j_2,n)]$ 
    \If{$sig_2=(|\hat{V}|-1)|\hat{V}|$} $\gamma\gets 0$ 
    \Else \  $\gamma \gets |MC(G_D(T(i_1,i_2, m),\hat{V}))|$ 
    \EndIf
    \If{$k-\gamma\geq 2K+1$} $sig \gets sig+1$ \Else{\ Break} \EndIf
    \EndFor
     \If{$sig=M$} \Return $\hat{V}$ \EndIf
    \EndFor
    \EndFor
   \end{algorithmic}
    \label{al:exp}
    \end{algorithm}
    \end{minipage}
}

\begin{figure}[t]
\begin{minipage}[b]{0.45\textwidth}
\centering
\begin{tikzpicture}
[> = stealth, 
	shorten > = 0.1pt, 
	auto,
	node distance = 3cm, 
	semithick 
	,scale=.8,auto=left,every node/.style={circle,draw, inner sep=2pt}]
\node (n1) at (0,1)		{$v_1$};
	\node (n2) at (1,0)  	{$v_2$};
	\node (n3) at (0.6,-1) 	{$v_3$};
	\node (n4) at (-0.6,-1) 	{$v_4$};
	\node (n5) at (-1,0)    {$v_5$};
\end{tikzpicture}
\end{minipage}
\begin{minipage}[b]{0.45\textwidth}
\centering
\begin{tikzpicture}
[> = stealth, 
	shorten > = 0.1pt, 
	auto,
	node distance = 3cm, 
	semithick 
	,scale=.8,auto=left,every node/.style={circle, draw, inner sep=2pt}]
\node (n1) at (0,1)		{$v_1$};
	\node (n2) at (1,0)  	{$v_2$};
	\node (n3) at (0.6,-1) 	{$v_3$};
	\node (n4) at (-0.6,-1) 	{$v_4$};
	\node (n5) at (-1,0)    {$v_5$};
	\draw (n1)--(n4);
\end{tikzpicture}
\captionsetup{font=footnotesize}
\end{minipage}
\captionof{figure}{(a) $G_D(T(1,2,3),V)$ (left) and (b) $G_D(T(2,3,3),V)$ (right) }
\label{fig:example1}
\end{figure}
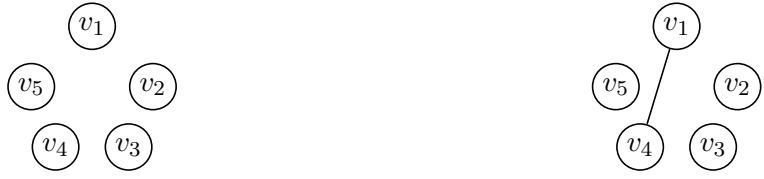

\begin{example}\label{example1}
Let $Att_1=[00001,11100]$, $Att_2=[10100]$, $Att_3=[11001]$ and $V=\{v_1, v_2, v_3, v_4, v_5\}$.
Part of $D$ is shown below, where gene pairs with $v_5$ are omitted.
$G_D(T(1, 2, 3), V)$ and $G_D(T(2, 3, 3),$ $V)$  are shown in Fig. \ref{fig:example1}.\\
{\small
$$
D= 
\left(
\begin{array}{cccccc|l}
(v_1, v_2) & (v_1, v_3) & (v_1, v_4) & (v_2, v_3) & (v_2, v_4) & (v_3, v_4) &\\
\hline
2 & 2 & 1 & 2 & 1 & 1 &(Att_1,Att_2)\\
2 & 2 & 1 & 2 & 1 & 1 &(Att_1, Att_3)\\
1 & 1 & 0 & 2 & 1 & 1 & (Att_2, Att_3)
\end{array}
\right),
$$
}
{\bf Case 1}: Consider $(Att_2, Att_3)$ and $\hat{V}=\{v_1, v_2, v_4\}$, then $\hat{E}=\{(v_1, v_4)\}\neq \emptyset$.
From Lemma \ref{lem:clique}, $Dist(Att_2, Att_3$, $\hat{V})=3-|MC(G_D(T(2,3,3),\hat{V}))|=1$ holds.
This can be verified since $Dist(Att_2, Att_3, \hat{V})$ equals to the Hamming distance between vectors $100$ and $110$.\\
{\bf Case 2}: Consider $(Att_1, Att_2)$ and $\hat{V}=\{v_1, v_2, v_3\}$, then $\hat{E}=\emptyset$.
In this case, $Dist(Att_1, Att_2, \hat{V})=3$ since these two attractors are discriminated by observing any node in $\hat{V}$,
which corresponds to Line 7 in Algorithm \ref{al:exp}.\\
{\bf Case 3}: Consider $(Att_2, Att_3)$ and $\hat{V}=\{v_1, v_2, v_3\}$, then $\hat{E}=\emptyset$.
In this case, $Dist(Att_2, Att_3, \hat{V})=3-|MC(G_D(T(2,3,3),\hat{V}))|=2$ since these two attractors are the same by observing $v_1$ but are discriminated by observing $v_2$ or $v_3$, which corresponds to Line 8 in Algorithm \ref{al:exp}.\\
Finally,
we can see that $\hat{V}=\{v_1, v_2, v_3, v_5\}$ is a solution of {\bf MinDattNN} for $K=1$ because the distance between any attractor pair is 3 or 4.
\end{example}

\begin{theorem}
 Algorithm \ref{al:exp} computes an optimal discriminator $\hat{V}^*$ in
$O(3^{|\hat{V}^*|/3}n^{|\hat{V}^*|}m^2+m^2 n^2 p)$ time,
where $p=LCM\{p(1), p(2), \ldots, p(m)\}$.
\label{thm:exact}
\end{theorem}
\begin{proof}
The correctness of the algorithm follows from Lemma~\ref{lem:clique}.
$D$ can be calculated in $O(m^2 n^2 p)$ time
by a naive algorithm.
We can apply the Bron-Kerbosch algorithm
to calculate the maximum clique,
whose time complexity is $O(3^{n/3})$ for graphs with $n$ nodes.
Therefore, the total time complexity
is $O(3^{|\hat{V}^*|/3}n^{|\hat{V}^*|}m^2+m^2 n^2 p)$,
where $|\hat{V}^*|$ is the minimum number of needed nodes.
\end{proof}

\subsection{Approximation Algorithm for {\bf MinDattNN}}

Algorithm \ref{al:app} is
a greedy-type
approximation algorithm running in $O(poly(m,n,p))$ time,
which is much more efficient than Algorithm \ref{al:exp}.
Here, we introduce some notations.
$x_{T(i_1,i_2,m)}$ means a POA $(Att_{i_1}, Att_{i_2})$, and $s_{T(j_1, j_2,n )}=\{x_{T(i_1, i_2, m)}|D[T(i_1, i_2, m), T(j_1, j_2, n)]\neq 0\}$. For instance, $x_1$ means $(Att_1, Att_2)$ and $s_1=\{x_1, x_2, x_3\}$ in Example \ref{example1}.
Notice that we adopt $r_{T(i_1, i_2, m)}$ to record the distance between $Att_{i_1}$ and $Att_{i_2}$,
and $r_{T(i_1, i_2, m)}\leq Dist(Att_{i_1}, Att_{i_2}, \hat{V})$ holds where
$\hat{V}$ is the current discriminator set.
In this greedy algorithm, a pair of nodes is added in each iteration, instead of a single node.
It is because there exist cases in which the single node addition
strategy fails
(Proposition~\ref{prop:nofeasible}),
whereas the node pair addition strategy can always find a feasible solution 
(Proposition~\ref{prop:add2}). 
An illustrative example is given in Example 2.

{\centering
\begin{minipage}{0.9\textwidth}
 \begin{algorithm}[H]
    \caption{Approximation algorithm for {\bf MinDattNN}}
    \hspace*{\algorithmicindent} \textbf{Input:}
$U=\{x_1,x_2,\ldots, x_{\binom{m}{2}}\}$, $S=\{s_1, s_2, \ldots, s_{\binom{n}{2}}\}$, integer $K$\\
    \hspace*{\algorithmicindent} \textbf{Output:}
set of nodes $\hat{V}$
    \begin{algorithmic}[1]
    \State Calculate matrix $D$
 \State $\hat{V}\gets \emptyset, r_l=0$ for $1\leq l \leq M$
\While{$U\neq \emptyset$ and $S\neq \emptyset$}
\State Find $s_{T(j_1, j_2, n)}\in S$ maximize $|s_{T(j_1, j_2, n)}\cap U|$
\State $\hat{V} \gets \hat{V}\cup \{v_{j_1}, v_{j_2}\}$, 
\State $S\gets S-\{s_{T(k, k', n)} | k=j_1 ~\text{or}~
k'=j_2\}$;
\For{all $x_l \in s_{T(j_1,j_2, n)}$ and $r_{l}<2K+1$}
\State $r_l\leftarrow r_l+D(l,T(j_1,j_2,n))$
\If{$r_l\geq 2K+1$} $U\leftarrow U-\{x_l\}$ \EndIf
\EndFor
\EndWhile
\If{$S=\emptyset$} $\hat{V}\gets V$ \EndIf
    \end{algorithmic}
    \label{al:app}
    \end{algorithm}
    \end{minipage}
    }
    
\begin{proposition}
Suppose that $(Att_{i_1},  Att_{i_2})$ is a pair of attractors and
$\hat{V}$ is a discriminator set such that
$Dist(Att_{i_1}, Att_{i_2}, \hat{V}) =N$.
Then it is possible that for any node
$v_j\in V-\hat{V},Dist(Att_{i_1}, Att_{i_2}, \hat{V}\cup \{v_j\})=N$,
but there exists $v_i \in \hat{V}$ such that
$Dist(Att_{i_1}, Att_{i_2}, \{v_i, v_j\})\neq 0$.
\label{prop:nofeasible}
\end{proposition}
\begin{proof}
Let $Att_1=[00101,00110], Att_2=[00010,00001]$, and $\hat{V}=\{v_1, v_2, v_3\}$.
Then we have
$Dist(Att_1, Att_2, \hat{V})=Dist(Att_1, Att_2, \{v_1, v_2, v_3, v_4\})
=Dist(Att_1, Att_2, \{v_1,$
$v_2, v_3, v_5\})=1$.
On the other hand, we have
$Dist(Att_1,$
$Att_2, \{v_3, v_4\})=Dist(Att_1, Att_2, \{v_3, v_5\}) \neq 0.$
\end{proof}

\begin{proposition}
If $(Att_{i_1},  Att_{i_2})$ is a pair of attractors and $\hat{V}$ is a discriminator such that $Dist(Att_{i_1},$ $Att_{i_2}, \hat{V}) = N$,
then for any pairs of nodes $(v_{j_1}, v_{j_2})$ such that
$v_{j_1}\notin \hat{V}, v_{j_2} \notin \hat{V}$, and $x_{T(i_1, i_2, m)}\in s_{T(j_1, j_2, n)}$,
$Dist(Att_{i_1}, Att_{i_2}, \bar{V}) \geq N+D[T(i_1, i_2, m), T(j_1, j_2, n)]$
holds where $\bar{V}=\hat{V}\cup \{v_{j_1}, v_{j_2}\}$.
\label{prop:add2}
\end{proposition}

\begin{proof}
Let $Dif_{(i_1, i_2)}(\hat{V}, t)$ has the form of 
$
Dif_{(i_1, i_2)}(\hat{V}, t)=\underbrace{00\cdots0}_{\hat{V}_1(t)}\underbrace{11\cdots1}_{\hat{V}_2(t)}.  
$
and\\
$
\beta_1(t)=\displaystyle\sum_{v_i\in \hat{V}} {I\{Ser(Att_{i_1},\{v_i\}, 0), Ser(Att_{i_2}, \{v_i\}, t)\}}.
$
Then, it is obvious $\beta_1(t)=|\hat{V}_2(t)|$. In the proof of Lemma~\ref{lem:clique},
we claimed a relationship between $\beta_1(t)$ and $Dist(Att_{i_1}, Att_{i_2}, \hat{V})$ that
$
Dist(Att_{i_1}, Att_{i_2},\hat{V})=\displaystyle\min_t{\beta_1(t)}.
$
Let $\alpha=D[T(i_1, i_2, m), T(j_1, j_2, n)]$.
Since $x_{T(i_1, i_2, m)}\in s_{T(j_1, j_2, n)}$, $\alpha$ is either 1 or 2.
This means
\begin{eqnarray*}
\alpha = \min_{t} \alpha(t) ~=~\min_{t} \hspace{5cm}\\
\left\{ \sum_{v_i\in \{v_{j_1}, v_{j_2}\}} {I\{Ser(Att_{i_1},\{v_i\}, 0), Ser(Att_{i_2}, \{v_i\}, t)\}} \right\}.
\end{eqnarray*}
Similarly, consider
$
\beta_2(t)=\displaystyle\sum_{v_i\in \bar{V}} {I\{Ser(Att_{i_1},\{v_i\}, 0), Ser(Att_{i_2}, \{v_i\}, t)\}}
$
then
$
Dist(Att_{i_1}, Att_{i_2}, \bar{V})
$
$
=\displaystyle\min_{t} \beta_2(t).
$
Obviously, we have
$\displaystyle\min_{t}\beta_2(t)\geq \min_{t}\beta_1(t)+\alpha$, and
the equality is achieved when both $\beta_1(t)$ and $\alpha(t)$ are minimized
at the same $t$, which implies
$
Dist(Att_{i_1}, Att_{i_2}, \bar{V}) \geq N+D[T(i_1, i_2, m), T(j_1, j_2, n)].
$
\end{proof}

\begin{example}
Three attractors are given by $Att_1=[010101,011011,000101,$ $111011,110101,101011]$, $Att_2=[010011,011100,000011,111100,110011,$ $101100]$ and 
$Att_3=[010001,101110,110001,001110]$.
A detailed process
of Algorithm \ref{al:app} is shown in
Table \ref{tab:n7} where ${\bf r}=[r_1, r_2,r_3]$.
After the 3rd step, all nodes have been added to $\hat{V}$ and
${\bf r}=[3,3,3]$ holds,
and then it returns $V$ since $U=\emptyset$.

\begin{table}
\caption{Example of execution of Algorithm 2.}
\begin{center}
\resizebox{0.9\textwidth}{!}{
\begin{tabular}{|c |c c c c| }
\hline
{\bf step}& $U$ & $S$ & $\hat{V}$ &
${\bf r}=[r_1,r_2,r_3]$  \\
\hline
0 & $\{x_1, x_2, x_3\}$ & $\{s_1, \ldots, s_{15}\}$&$\emptyset$& ${[0,0,0]}$\\
1 & $\{x_1, x_2, x_3\}$ & $\{s_6, s_8, s_9, s_{11}, s_{12}, s_{15}\}$&$\{v_1, v_4\}$& ${[1,1,1]}$  \\
2 & $\{x_1,x_3\}$ & $\{s_{11}\}$&$\{v_1, v_4, v_2, v_6\}$& ${[2,3,2]}$  \\
3 & $\emptyset$ & $\emptyset$&$\{v_1, v_4, v_2, v_6, v_3, v_4\}$& ${[3,3,3]}$  \\
\hline
\end{tabular}
}
\end{center}
\label{tab:n7}
\end{table}
\end{example}    
 Besides,
there is a difficulty in analyzing the approximation factor
in general.
Therefore, we add a condition 
$Dist(Att_{i_1}, Att_{i_2}, \hat{V}_{j-1}) +Dist(Att_{i_1},$ $Att_{i_2}, \hat{V}^*-\hat{V}_{j-1})\geq 2K+1$
to obtain a guaranteed approximation factor as below.
Even though it is difficult to test whether or not the condition is satisfied
before running the approximation algorithm,
it seems from numerical experiments
that this condition is satisfied in most cases.

\begin{theorem}
Let $\hat{V}_{j}$ denote the discriminator set obtained after the $j$-th
iteration of Algorithm \ref{al:app}
and $\hat{V}^*$ be an optimal solution.
Suppose that
$
Dist(Att_{i_1}, Att_{i_2}, \hat{V}_{j-1}) +Dist(Att_{i_1}, Att_{i_2}, \hat{V}^*-\hat{V}_{j-1}) \geq 2K+1
$
is satisfied for each $j$.
Then, Algorithm \ref{al:app} is an $\ln(M(2K+1))+1$ factor
polynomial time approximation algorithm for {\bf MinDattNN} ($M=\binom{m}{2}$).
\label{thm:apsolve}
\end{theorem}
\begin{proof}
Firstly, we consider the approximation ratio when the algorithm terminates when $S\neq \emptyset$.
Let $t(k)$ be the index such that a pair of nodes $s_{t(k)}$ is chosen at the $k$-th iteration.
Let $\hat{V}_{j-1}$ denote the discriminator set after the $(j-1)$-th iteration, $U_j$ is $U$ at the $j-$th iteration.
We assume without loss of generality that $\hat{V}_{j-1}=\{v_1, v_2, \cdots, v_{2j-2}\}$ holds.
In each iteration, a pair of nodes is added, thus
the distance
between a POA can be increased by at most 2, thus we have the following inequality:
\begin{eqnarray}
\label{ineq:G2}
&\sum_{i_1< i_2}{Dist(Att_{i_1},Att_{i_2}, \hat{V}_{j-1})} \leq \sum_{i_1<i_2}(0+2\cdots+2)&\\
&=2 \sum_{k=1}^{j-1} \sum_{i_1<i_2}
{\bf 1}_{s_{t(j)}\cap U_k}(x_{T(i_1, i_2, m)})
= 2 \sum_{k=1}^{j-1}|s_{t(k)}\cap U_k|,& \nonumber
\end{eqnarray}
where ${\bf 1}_A(x)=1$ if $x \in A$, and 0 otherwise.\\
By assumption, we have
$Dist(Att_{i_1},Att_{i_2},\hat{V}_{j-1})+Dist(Att_{i_1}, Att_{i_2}, \hat{V}^*-\hat{V}_{j-1}) \geq 2K+1$.
By taking the summation over all pairs of attractors, we have
\begin{eqnarray}
\label{ineq:G3}
& & \sum_{i_1<i_2}Dist(Att_{i_1},Att_{i_2},\hat{V}^*-\hat{V}_{j-1})\\
& \geq & (2K+1)M-\sum_{i_1<i_2}{Dist(Att_{i_1},Att_{i_2},\hat{V}_{j-1})}\nonumber\\
& \geq & (2K+1)M-2\displaystyle\sum_{k=1}^{j-1}|s_{t(k)}\cap U_k|. \nonumber
\end{eqnarray}
At the $j$-th iteration, we need to choose a pair of nodes that 
discriminates the largest number of POAs in $U_j$.
Since all nodes in $\hat{V}^*-\hat{V}_{j-1}=\{v_1^*,v_2^*,\cdots\}$ are candidates,
every two nodes form a pair in order, which
are denoted as $s^*_1, s^*_2, \cdots, s^*_{|\hat{V}^*-\hat{V}_{j-1}|/2}$.
Then we have
\begin{eqnarray}
\label{ineq:G4}
\begin{array}{l}
\displaystyle\sum_{k=1}^{|\hat{V}^*-\hat{V}_{j-1}|/2} |s_k^*\cap U_j| \leq \frac{|\hat{V}^*-\hat{V}_{j-1}|}{2} \cdot |s_{t(j)}\cap U_j|,
\end{array}
\end{eqnarray}
Furthermore, if we consider choosing nodes from $\hat{V}^*-\hat{V}_{j-1}$
to discriminate POAs in $U_{j-1}$,
then after $|\hat{V}^*-\hat{V}_{j-1}|/2$ iterations, all nodes will
be chosen.
Therefore, from Ineq. (\ref{ineq:G2}), we have
\begin{eqnarray}
\label{ineq:G5}
&  & \sum_{k=1}^{|\hat{V}^*-\hat{V}_{j-1}|/2} |s_k^*\cap U_j| 
\geq \displaystyle\sum_{k=1}^{|\hat{V}^*-\hat{V}_{j-1}|/2} |s_k^*\cap U^k_j| \\
& \geq & \frac{1}{2} \displaystyle\sum_{i_1<i_2}Dist(Att_{i_1},Att_{i_2}, \hat{V}^*-\hat{V}_{j-1}),\nonumber
\end{eqnarray}
where $U_j^k$ is $U_j$ at the $k$-th iteration. Putting
{
(\ref{ineq:G3})-(\ref{ineq:G5})
}
together, we have
\begin{align*}
\frac{1}{|s_{t(j)}\cap U_j|} & \leq \frac{ |\hat{V}^*-\hat{V}_{j-1}|}{2 \displaystyle\sum_{k=1}^{|\hat{V}^*-\hat{V}_{j-1}|/2} |s_k^*\cap U_j|} & \text{by (3)}\\
&\leq \frac{|\hat{V}^*-\hat{V}_{j-1}|}{\displaystyle\sum_{i_1<i_2}Dist(Att_{i_1},Att_{i_2}, \hat{V}^*-\hat{V}_{j-1})} &\text{by (4)}\\
&\leq \frac{|\hat{V}^*-\hat{V}_{j-1}|}{M(2K+1)-2 \displaystyle\sum_{k=1}^{j-1}{|s_{t(k)}\cap U_k |}} &  \text{by (2)}\\
&\leq \frac{|\hat{V}^*|}{M(2K+1)-2 \displaystyle\sum_{k=1}^{j-1}{|s_{t(k)}\cap U_k |}}. &
\end{align*}
Finally, 
we define $P(x_l, j)$ as the average price of 
to each $x_l$ at the $j$-th iteration, for $x_l\in U_j$.
Then we have
$
P(x_l, j)=\left \{
\begin{array}{cl}
\frac{2}{|s_{t(j)}\cap U_j|}, & x_l \in s_{t(j)}\cap U_j,\\
0, & {\rm otherwise}.
\end{array}
\right.
$\\
After termination of Algorithm \ref{al:app},
the total number of nodes in the discriminator should be $|\hat{V}|$,
which means
$
|\hat{V}|=\displaystyle\sum_{j=1}^{|\hat{V}|/2}{ \sum_{l=1}^{M} P(x_l, j)}
$.
Thus we have
\begin{eqnarray*}
|\hat{V}|&=&\sum_{j=1}^{|\hat{V}|/2}{ \sum_{l=1}^{M} P(x_l, j)}\\
& \leq & \sum_{j=1}^{|\hat{V}|/2}(2|U_j\cap s_{t(j)}|)
\cdot \frac{|\hat{V}^*|}{M(2K+1)-2 \displaystyle\sum_{k=1}^{j-1}|s_{t(k)}\cap U_k|}\\
& \leq & |\hat{V}^*|\left(1 + \frac{1}{2}+\ldots + \frac{1}{M(2K+1)}\right),
\end{eqnarray*}
where
$i_s=2\displaystyle\sum_{k=1}^{j-1}|s_{t(k)}\cap U_k|$ and
$i_e=2\displaystyle\sum_{k=1}^{j}|s_{t(k)}\cap U_k|-1$.
Therefore, we have\\
$
\frac{|\hat{V}|}{|\hat{V}^*|} \le 
\left(1 + \frac{1}{2}+\ldots + \frac{1}{M(2K+1)}\right)
\leq \ln(M(2K+1))+1.
$
On the other hand, if the algorithm terminates when $S=\emptyset$,
it will reduce the approximation factor.
Therefore, $\ln{(M(2K+1))}+1$ gives an upper bound of
the approximation ratio.\\
Next, we analyze the time complexity.
Let $p=LCM\{p(1), p(2), \ldots, p(m)\}$.
{Calculation of matrix $D$ should cost $O(m^2n^2p)$ time.}
It is also clear that
the inner {\bf While} loop takes $O(m^2 n^2 \log{n} )$ time.
Therefore, the theorem holds.
\end{proof}

\section{Discrimination of Singleton Attractors with Noisy Nodes}

As we mentioned before, discrimination of singleton attractors with noisy nodes is a special case of the discrimination problem proposed in the previous section.
In this case, each attractor $Att_i$ is a binary vector and
the distance between a POA $(Att_{i_1}, Att_{i_2})$ is degenerated to $H(Att_{i_1}, Att_{i_2})$,
where $H(x, y)$ is the Hamming distance between binary vectors $x$ and $y$.
Therefore, we put all attractors in a binary $m\times n$ matrix $Att$ in which each row represents a singleton attractor. Moreover, for a matrix $A$, let $A[i,-]$ (resp., $A[-,j]$) denotes the $i$-th row (resp., $j$-th column).
Similarly, for a matrix $A$ and a set of column (resp., row) indices $J=\{j_1,\ldots,j_k\}$,
$A[i,J]$ (resp., $A[J,j]$) denotes the submatrix of $A[i,-]$ (resp., $A[-,j]$)
consisting of the $j_1,\ldots,j_k$-th columns (resp., rows).
\cite{Cheng2017b} considered the clean (i.e., without noise) version of
the problem,
The task was to find the minimum set of column indices $J$ such that
$Att[i_1, J]\neq Att[i_2, J]$ holds for all $i_1,i_2$ with
$1 \leq i_1 < i_2 \leq m$.
In the noisy case, it is hypothesized that noisy nodes vary in
different attractors, thus at most $2K$ nodes are not reliable for a POA, 
thus a POA is discriminated only if $H(Att[i_1, -], Att[i_2, -])>2K$.
Hence, we define {\bf MinDSattNN} as follows.
%
\begin{definition}
(Minimum Discriminator for Singleton Attractors with Noisy Nodes 
{\bf [MinDSattNN]})\\
\noindent
{\bf Input}: A set of singleton attractors represented by an
$m\times n$ binary matrix $Att$ and an integer $K$ denoting
the maximum number of noisy nodes per attractor. \\
{\bf Output}: A minimum cardinality set $J$ of columns such that
$H(Att[i_1, J],$ $Att[i_2, J]) \geq 2K+1$
holds for all $i_1, i_2$ with $1 \leq i_1< i_2 \leq m$.
\end{definition}

\subsection{Exact Algorithm for {\bf MinDSattNN}}
To solve {\bf MinDSattNN}, we develop an integer programming (IP)-based
exact method.
To this end, we construct a matrix $C_{att}$
of size $\binom{m}{2} \times n$ by
comparing every two rows of $Att$:
$$
 C_{att}[T(i_1, i_2, m), j]=\left\{
 \begin{array}{ll}
 0, & {\rm if} \ Att[i_1, j]= Att[i_2, j],\\
 1, & {\rm if} \ \mbox{otherwise}.
 \end{array}
 \right.
 $$
In {\bf MinDatt}, matrix $D$ is constructed from calculating distance between any POA by observing a pair of nodes, whereas in the singleton case observing only a node is needed.
Let $M$ denote $\binom{m}{2}$.
Let ${\bf y} = [y_1, y_2, \ldots, y_n]$ be a vector in which
$y_j$ takes 1 ($j \in J$) or 0 ($j \notin J$).
Then {\bf MinDSattNN} can be formulated
as a typical IP problem
$$ \min \  {\bf 1\cdot y}^T $$
subject to
$$
\left\{
\begin{array}{l}
C_{att}[i, -] \cdot {\bf y}^T \geq 2K+1 \quad (i=1, 2, \ldots, M),\\
y_j \in \{0,1\}\quad (j=1, 2, \ldots, n),
\end{array}
\right.
$$
where ${\bf 1}=\underbrace{[1,1, \ldots, 1]}_{n}$.
Accordingly, we see that
{\bf MinDSattNN} can be transformed into an integer programming problem
with $n$ binary variables and $O(m^2)$ constraints.
It should be noted that
existing IP solvers (e.g., {\it intlinprog} in MATLAB)
take exponential time in the worst case
and thus this IP-based method takes exponential time in the worst case.
However, it is reasonable because both {\bf MinDattNN} and {\bf MinDSattNN}
include the problem of discrimination of singleton attractors
({\bf MinDiscSatt}) \cite{Cheng2017b}, which is known to be NP-hard, as a special case
and thus are NP-hard.
The following is an illustrative example of the IP process.
\begin{example}
Let $Att$ be given by
$$
Att = 
\left(
\begin{array}{llllllll}
1 & 0 & 0 & 0 & 0 & 0 & 0 & 1\\
1 & 1 & 1 & 0 & 1 & 0 & 0 & 1\\
1 & 0 & 0 & 0 & 1 & 1 & 1 & 0
\end{array}
\right),
$$
which means that there exist three singleton attractors
$Att[1,-] = [1,0,0,0,$ $0,0,0,1]$,
$Att[2,-] = [1,1,1,0,1,0,0,1]$,
$Att[3,-] = [1,0,0,0,1,1,1,0]$.
Then we would have
$$
C_{att} = 
\left(
\begin{array}{llllllll}
0 & 1 & 1 & 0 & 1 & 0 & 0 & 0\\
0 & 0 & 0 & 0 & 1 & 1 & 1 & 1\\
0 & 1 & 1 & 0 & 0 & 1 & 1 & 1
\end{array}
\right),
$$
and all parameters into {\it intlinprog} (in MATLAB)
would be $f=[1,1,1,1,1,1,$ $1,1]$, $intcon=[1,2,3,4,5,6,7,8]$,$lb=[0,0,0,0,0,0,0,0]$, $ub=[1,1,1,1,1,$ $1,1,1]$,
then we would have the optimal object value of 5 and
$y=[0,1,1,0,1,$ $1,0,1]$, which means $J=\{2,3,5,6,8\}$.
\label{example:ip}

\end{example}

\subsection{Approximation Algorithm for {\bf MinDSattNN}}

{\bf MinDSattNN} is a special case of
the set multi-cover problem \cite{rajagopalan1993primal}, which is NP-hard.
Therefore,
in order to balance the trade-off between the size of a target set $J$ and the overall time complexity,
we design a simple greedy algorithm, Algorithm~\ref{al:aps},
based on \cite{rajagopalan1993primal}.
As shown in Theorem~\ref{thm:singleton},
it has a guaranteed approximation ratio $\displaystyle \ln{\left(M(2K+1)\right)+1}$.
Notice that, usually, $m\ll n$ and $K \ll n$ and thus
the ratio is acceptable as
the computational time can be reduced significantly.
Similarly, $x_{T(i_1,i_2, m)}$ denotes a POA and let $s_j=\{x_{T(i_1, i_2, m)}| Att[i_1, j]\neq Att[i_2, j]\}$ denote the POAs that can be discriminated by $v_j$. For example, $x_1=(Att_1, Att_2)$ and $s_2=\{x_1, x_3\}$ in Example~\ref{example:ip}. 
An example of this algorithm is given in Example~\ref{example:aps}.

{\centering
\begin{minipage}{0.9\textwidth}
 \begin{algorithm}[H]
    \caption{Approximation algorithm for {\bf MinDSattNN}}\label{APsin}
    \hspace*{\algorithmicindent} \textbf{Input:} $U=\{x_l, 1\leq l \leq M\}$, $S=\{s_1, s_2, \ldots, s_n\}$, integer $K$ \\
    \hspace*{\algorithmicindent} \textbf{Output:} set of nodes $J$
    \begin{algorithmic}[1]
    \State $J\gets \emptyset, r_l=0$ for $1\leq l \leq M$
    \While{$U\neq \emptyset$}
    \State Find { $s_j\in S$} with maximum $|s_j\cap U|$,
    $J\gets J\cup \{j\}, S\gets S-\{s_j\}$;
    \For{all $x_l \in s_j$ and $r_{l}<2K+1$}
    $r_l\gets r_l+1$
    \If{$r_l=2K+1$} $U\leftarrow U-\{x_l\}$ \EndIf
    \EndFor
    \EndWhile
    \end{algorithmic}
    \label{al:aps}
    \end{algorithm}
    \end{minipage}
    }
    
 \begin{example}
    Attractors are the same as in Example \ref{example:ip},
and detailed execution steps are shown in Table ~\ref{tab:n12}. 
In this case, we have the final target set
equals to $\{v_2, v_3, v_5, v_6, v_7\}$.
\begin{table}[ht]
\caption{Example of execution of Algorithm~\ref{al:aps}.}
\begin{center}
\resizebox{0.9\textwidth}{!}{
\begin{tabular}{|c |c c c c| }
\hline
{\bf step}& $U$ & $S$ & $J$ & ${\bf r}=[r_1,r_2,r_3]$  \\
\hline
0 & $\{x_1, x_2, x_3\}$ & $\{s_1, \ldots, s_{8}\}$&$\emptyset$&
${[0,0,0]}$  \\
1 & $\{x_1, x_2, x_3\}$ & $\{s_1, s_3, \ldots, s_{8}\}$&$\{2\}$& ${[1,0,1]}$  \\
2 & $\{x_1,x_2, x_3\}$ & $\{s_1, s_4, \ldots, s_{8}\}$&$\{2,3\}$& ${[2,0,2]}$  \\
3 & $\{x_2, x_3\}$ & $\{s_1, s_4, s_6, s_7, s_{8}\}$&$\{2,3,5\}$& ${[3,1,2]}$  \\
4 & $\{x_2\}$ & $\{s_1, s_4, s_7, s_{8}\}$&$\{2,3,5,6\}$& ${[3,2,3]}$  \\
5 & $\emptyset$ & $\{s_1, s_4, s_{8}\}$&$\{2,3,5,6,7\}$& ${[3,3,3]}$  \\
\hline
\end{tabular}
}
\end{center}
\label{tab:n12}
\end{table}
\label{example:aps}
\end{example}

    \begin{theorem}
Algorithm~\ref{al:aps} is an $\ln(M(2K+1))+1$ factor
polynomial-time
approximation algorithm for {\bf MinDSattNN}.
\label{thm:singleton}
\end{theorem}
\begin{proof}
Let $t(j)$ be the index such that $s_{t(j)}$ is chosen at the $j$-th iteration.
We say that price for solving {\bf MinDSattNN} is $N$ if
$N$ nodes are needed for discriminating the given attractors each with
at most $K$ noisy nodes.
In each iteration, a column index is added to $J$ and then 
price 1 is added to the total cost.
Afterward, we assign an average price of $P(x_l, j)$ to each $x_l$
at the $j$-th iteration, for $x_l\in U_j$,
where $U_j$ is the set $U$ at $j$-th iteration.
Therefore we have
$
P(x_l, j)=\left \{
\begin{array}{cl}
\frac{1}{|s_{t(j)}\cap U_j|}, & x_l \in s_{t(j)}\cap U_j,\\
0, & {\rm otherwise}.
\end{array}
\right.
$

We call $|s_{t(j)}\cap U_j|$  {\it coverage power} of $s_{t(j)}$.
After termination of Algorithm~\ref{al:aps},
the total number of nodes in the discriminator set should be $|J|$
from which 
$
|J|=\sum_{j=1}^{|J|}{ \sum_{l=1}^{M} P(x_l, j)}
$
holds.
Here we give a key inequality (A.1),
$
P(x_l, j)\leq
\displaystyle \frac{|J^*|}{M(2K+1)-\displaystyle \sum_{k=1}^{j-1}|s_{t(k)}\cap U_k|},
$
where we define the summation in the denominator to be 0 when $j=1$
because $s_{t(0)}=\emptyset$, and
the proof is given in Appendix~\ref{sec;keyeq}.

Let $J^*$ be an optimal solution of {\bf MinDSattNN}.
Then, we have
\begin{eqnarray*}
|J|&=& \sum_{j=1}^{|J|}{ \sum_{l=1}^{M} P(x_l, j)}
 = \sum_{j=1}^{|J|}(|U_j\cap s_{t(j)}|) P(x_l, j)\\
& \leq & \sum_{j=1}^{|J|}(|U_j\cap s_{t(j)}|)
\cdot \frac{|J^*|}{M(2K+1)-\displaystyle\sum_{k=1}^{j-1}|s_{t(k)}\cap U_k|}\\
& \leq & |J^*|\displaystyle\sum_{j=1}^{|J|}\left(\displaystyle\sum_{s=i_s}^{i_e}\frac{1}{M(2K+1)-s}\right )\\
& \leq & |J^*|\left(1 + \frac{1}{2}+\ldots + \frac{1}{M(2K+1)}\right),
\end{eqnarray*}
where $i_s=\sum_{k=1}^{j-1}|s_{t(k)}\cap U_k|$ and 
$i_e=\sum_{k=1}^{j}|s_{t(k)}\cap U_k|-1$,
and the second inequality comes from
$
\sum_{k=1}^{|J|-1}{|s_{t(k)}\cap U_k|}\geq 1
$.
Therefore, the approximation ratio is bounded by

$\displaystyle{
\frac{|J|}{|J^*|} \le 1 + \frac{1}{2}+\ldots + \frac{1}{M(2K+1)} 
< \ln(M(2K+1)) + 1.
}$

Hereafter we analyze the time complexity.
The number of iterations in {\bf While} loop is bounded by
$\min \left\{
|J^*|\cdot \ln (M (2K+1))+1, n\right\}.
$
and time complexity for each iteration of
{\bf While} loop is at most $O(m^2 \cdot n \log{n})$ by using merge sort.
Therefore, the total cost for Algorithm~\ref{al:aps}  is
$
O\left(m^2n\log{n} \cdot\right.$ 
$\left. \min
\left\{ |J^*| \cdot (\ln (M (2K+1))+1), n\right\}\right),
$
which is a polynomial order of $m$ and $n$.
\end{proof}

\section{Results of Computational Experiments}
\label{sec:biological}

All numerical experiments were conducted using Matlab on a PC with dual-core
3.4 GHz processor and 8 GB RAM.
Firstly, for each of {\bf MinDattNN} and {\bf MinDSattNN},
we conducted computational experiments using simulation data by
randomly generating a couple of attractors,
repeating the numerical experiment 10 times for each parameter set,
and then recording the average time and maximum approximation ratio.
For the discrimination of attractors,
results are listed in Table \ref{tab:n9},
where $len$ denotes the maximum length of attractors.
It is seen that
Algorithm 2 is much faster than Algorithm 1 and
the approximation ratios are not large.
For discrimination of singleton attractors, it is seen from Table \ref{tab:n8} that 
Algorithm~\ref{al:aps} is much faster than the IP-based method,
and the approximation ratios are much smaller than 
$\displaystyle{\ln(M (2K+1)) + 1}$.

Next, we also examined the efficiency of
the IP-based method and Algorithm~\ref{al:aps}
for the clean case of {\bf MinDSattNN},
we compared those with a previous one, {\bf SolveMinDiscSatt} \cite{Cheng2017b}.
It is seen from Table~\ref{tab:n13} that
the new methods are much more efficient even for the clean case.

\begin{table}
\caption{Numerical results on discrimination of attractors.}
\begin{center}
\resizebox{0.8\textwidth}{!}{
\begin{tabular}{|c c c c c c c|}
\hline
\multirow{2}{*}{$n$}&\multirow{2}{*}{$m$} &  \multirow{2}{*}{$K$} & $\multirow{2}{*}{len}$ & Time of   & Time of  & Approximation  \\
& & & & {\bf Algorithm 1} (s) &  {\bf Algorithm 2} (s) &  ratio \\
\hline
100 & 3 &1 & 3 & 63.6160& 0.0100& 2\\
100 & 5 & 1 & 5 & 209.09 & 0.0100 &2 \\
100 & 5& 2& 5 & 260982 & 1.3&1.7\\
1000 & 3& 3 & 1 &9902.3 &0.0100 &2\\
\hline
\end{tabular}
}

\end{center}
\label{tab:n9}
\end{table}

\begin{table}
\caption{Numerical results on discrimination of singleton attractors.}
\begin{center}
\resizebox{0.8\textwidth}{15mm}{
\begin{tabular}{|c c c c c c|}
\hline
\multirow{2}{*}{$n$}&\multirow{2}{*}{$m$} &\multirow{2}{*}{$K$} &Time of & Time & Approximation\\
&&&IP(s) & {\bf Algorithm~\ref{al:aps}} (s)&ratio\\
\hline
50 & 5 &1 & 0.109 &0.003 & 1.3333 \\
50 & 5 & 3 & 0.045 & 0.002 & 1.1538 \\
500 & 5& 1 & 0.137& 0.001& 1.3333 \\
500 & 5 & 3 &  661.4& 0.000 &1.2000 \\
5000&5&3&7063& 0.000&1.2000 \\
20000&5&3& 7064& 0.080 &1.2000 \\
20000&5&5& 6855 & 0.017&1.1000 \\
20000&5&10& 7168& 0.400 & 1.1000 \\
\hline
\end{tabular}
}
\end{center}
\label{tab:n8}
\end{table}

\begin{table}
\caption{Numerical results on the clean case of discrimination of singleton attractors.}
\begin{center}
\resizebox{0.8\textwidth}{!}{
\begin{tabular}{|c c c c c c|}
\hline
\multirow{2}{*}{$n$}&\multirow{2}{*} {$m$} & Time of &Time of & Time of & Approximation\\
& & {\bf SolveMinDiscSatt} (s) &IP (s) & {\bf Algorithm~\ref{al:aps}} (s) &  ratio \\
\hline
100 & 5 &0.4094 & 0.0657 &0.0125 & 1 \\
1000 & 5 &7.2797 & 0.2703 &0.0031 & 1 \\
10000 & 5 &1571.2 & 51.6 &0.0000 & 1 \\
20000 & 5 &5128.2 & 125.5 &0.0000 & 1 \\

\hline
\end{tabular}
}
\end{center}
\label{tab:n13}
\end{table}

\begin{table*}[t]
\caption{Results on four biological processes.}
\label{tbl:biological}
\centering
\resizebox{0.8\textwidth}{12mm}{
\begin{tabular}{|c|cc|cc|cc|}
\hline
& $n$ & $m$ & \multicolumn{2}{|c|}{CPU time (sec)} & \multicolumn{2}{|c|}{Identified Markers}\\
& & & Exact & Approx. & Exact & Approx\\
\hline
\hline

(1) & 40 & {3+(7)}& 612.11 & 0.04 & ZAP70, TCR, SLP76, SEK, RLK &
ZAP70, TCR, SLP76, SEK, RLK, TCRphosp \\
\hline
\multirow{2}{*}{(2)} & {90} & {$(4,4,4,$} & $\approx 10^6$ & 0.3 & AFF1, AKAP12, APLP2, CAV1, & AEBP1, AFF1, AICDA, AKT3, APLP2 
\\ 
 & & $4,4,4)$ & & & CCND2, HDAC9, INPPSD & CAV1, CCND2, HDAC9, INPP5D, PAK\\
\hline
\hline
(3) & 60 & 5 & 0.52 & 0.53 & wg1, WG1, EN1, PTC1, PH1, ptc2, PTC2, SMO3  
& wg1, WG1, en1, EN1, hh1, en2, EN2, hh2, PTC2 \\
\hline
(4) & 9 & 3 & 0.05 & 0.06 & Rb, TELase, Cyclin, E2F, ESE2
& p53, p16, Rb, TELase, Snai2, E2F \\
\hline

\end{tabular}
}
\end{table*}

Then, we conducted computational experiments
using BN models on the following four real biological processes,
where $K=1$ was used in all cases:
(1) Logical model analyzing T-cell activation
(\cite{klamt2006methodology}),
(2) IGVH mutational status in chronic lymphocytic leukemia
\cite{alvarez2015proteins},
(3) Segment polarity genes in {\it Drosophila Melanogaster}
\cite{albert2003topology},
(4) Tumorigenic transformation
of human epithelial cells \cite{mendez2017gene}. 
Note that the algorithms for {\bf MinDattNN} were used for (1) and (2),
whereas those for {\bf MinDSattNN} were used for (3) and (4).
The results are summarized in Table~\ref{tbl:biological}.
In this table,
$n$ and $m$ denote the numbers of genes and attractors,
where the periodic attractors are shown by a list of their periods.

In the BN model (1),
there exist 9 periodic attractors \cite{klamt2006methodology}.
However, distances among some attractors are less than or equal to 2.
Therefore, we discarded such attractors.
Finally, we only kept attractors 1, 2, 3 and 9 for verification
It is known that ZAP70, TCR, SLP76, SEK, and RLK play
important roles in T-cell development and lymphocyte activation or development
of the nervous system, while genes ZAP70, TCR are the ligand for TCRphosp. Therefore, the approximation algorithm may have found more important genes.
In addition, the approximation algorithm was much faster than the exact
algorithm.
These facts suggest the usefulness of the approximation algorithm.

In the BN model (2),
the exact and approximation algorithms identified 7 and 10 genes,
respectively.
Therefore, the resulting approximation ratio is $10/7$,
which is much smaller than $\ln(M (2K+1)) + 1$.
Besides, there are 6 common genes identified,
which encode
proteins involved in critical biological processes or diseases
like human child lymphoblastic leukemia, histone deacetylase and so on.
Among the other identified genes,
AKAP12 gene functions in binding to the regulatory subunit of PAK and
confining the holoenzyme to discrete locations within the cell.
AEBP1 gene encodes proteins that may function as a transcriptional
repressor and play a role in adipogenesis and smooth muscle cell differentiation.
AICDA encodes a RNA-editing deaminase.
AKT3 encodes proteins known to be regulators of cell signaling in response to insulin and growth factors.
These facts suggest the usefulness of both algorithms.
However, the approximation algorithm was much faster than the exact one.
Therefore, these results suggest again
the usefulness of the approximation algorithm.

In the BN model (3),
there are 10 singleton attractors \cite{albert2003topology}.
However,
since the distances among attractors 3, 4, and 6 are less than 3,
and the distances among attractors 7, 8, 9, and 10 are also less than 3,
we only keep the five attractors: 1, 2, 3, 6, 7. {There are 60 nodes in this BN, including 5 segment polarity genes (en, wg, ptc, ci, hh) and their proteins (EN, WG, PTC, CI, CIR, CIA, SMO, HH), one pair-rule gene (slp) and its protein SLP in one parasegment primordia (4 cells). 
Most markers identified
by two algorithms are the same
except that gene hh is not identified
by the exact algorithm
and transcription factors PTC and SMO are not identified
by the approximation algorithm.
It has been verified experimentally that binding of HH (protein of hh) would remove the inhibition of SMO.} 
It is seen from Table~\ref{tbl:biological} that
the approximation ratio by Algorithm~\ref{al:aps}, 
is $9/8=1.125$, which is reasonable.  

The BN model (4) includes
9 molecular players (transcription factors or signaling molecules) \cite{mendez2017gene}.
After solving a system of Boolean equations, three attractors were obtained,
which corresponded to three kinds of cells:
{\it epithelial cells}, {\it senescent cells} and
{\it mesenchymal stem-like cells}.
The obtained molecular players can be interpreted as master regulators
for each cell.
It is known that cells with mesenchymal stem-like phenotype have a strong
potential of transferring to cinomas.
Snai2 was chosen as a master regulator from the approximation algorithm (but not from the IP-based method).
By taking a further look at this molecule,
we can see that the activation of Snai2 enables cells to sustain proliferative signals and to evade growth suppressors by undergoing a de-differentiation process.
Thus it is activated in {\it mesenchymal stem-like} cells
but not in the other two kinds of cells.
This fact suggests the usefulness of the approximation algorithm.

\section*{Acknowledgements}
XQC was partially supported by National Science Foundation of China under grant numbers 11801434 and 3115200128.
WKC was partially supported by Hong Kong RGC GRF Grant no. 17301519,
National Natural Science Foundation of China Under Grant number 11671158,
and IMR and RAE Research fund from Faculty of Science, HKU.
TA was partially supported by JSPS KAKENHI Grant number 18H04413.

\bibliographystyle{ieeetr}
\bibliography{discnoisybib}           

\appendix
\section{Proof of Key Inequality (A.1) in Theorem~\ref{thm:singleton}}    
\label{sec;keyeq}

Let $\alpha$ be the average price in the $j$-th iteration, that is,
$\alpha = P(x_l, j)$.
Let $J^*$ represent an optimal discriminator set.
When $j=1$, we need to prove
$
\alpha \leq \frac{|J^*|}{M(2K+1)}.
$
Note that in each iteration, we choose a node with the maximum coverage power (or the minimum average price), thus we have
$
\alpha \leq \frac{1}{|s_{k'}\cap U|}, \ \ k'\in J^*
$
which indicates
$
|J^*| \geq \alpha
\left( \sum_{k' \in J^*} |s_{k'}\cap U| \right).
$
Let ${\bf y}^*$ be the binary vector corresponding to $J^*$.
Then we have
$
\left( \sum_{k'\in J^*} |s_{k'}\cap U| \right) \geq M(2K+1)
$
from
$
C_{att}[i, -] \cdot ({\bf y}^*)^T \geq 2K+1\ \ \  (i=1,2,\ldots,M).
$
Thus the inequality is satisfied when $j=1$.

Before presenting our proof for the general case,
we give a definition of a special class of sets (multi-set)
by allowing it as a collection of objects
and any object may have duplications in the same set.
For example,  suppose
$S_1=\{x_1, x_1, x_2, x_3, x_3, x_3\}$ and
$S_2=\{x_1, x_3, x_4, x_4\}$.
Then $|S_i|$ is the cardinal number of the set $S_i$, which equals to the number of elements in $S_i$ and $|S_1|=6$ and $|S_2|=4$.
Moreover, we let
$
S_1\cup S_2=\{x_1, x_1, x_1,x_2, x_3, x_3, x_3, x_3, x_4, x_4\}
$
be the union of these two sets without removing duplications,
$
S_1\cap S_2=\{x_1, x_3\}
$
be the intersection of two sets by keeping the common elements of $S_1$ and $S_2$ with smaller frequency, and
$
S_1-S_2=\{x_1, x_2, x_3, x_3\}
$
be the relative complement of $S_1$ in $S_2$ by removing those common elements of $S_1$ and $S_2$ with smaller duplicates from $S_1$.
We will then apply this new definition in the following analysis.

Let
$\displaystyle{
U'=\{\underbrace{x_1, x_1, \ldots, x_1}_{2K+1}, \ldots, \underbrace{x_M, x_M, \ldots, x_M}_{2K+1}\}
}$
where $|U'|=M(2K+1)$.
Then the greedy algorithm can be rewritten as follows.

\setcounter{algorithm}{0}

{\centering
\begin{minipage}{0.9\textwidth}
 \begin{algorithm}[H]
    \caption{Greedy algorithm {\bf MAPMinDSattNN}}
    \hspace*{\algorithmicindent} \textbf{Input:} set of POAs $U'$, set of nodes $S$ \\
    \hspace*{\algorithmicindent} \textbf{Output:} set of nodes $J$
    \begin{algorithmic}[1]
    \State $J\gets \emptyset$ 
    \While{$U' \neq \emptyset$}
    \State Find { $s_j\in S$} with maximum $|s_j\cap U'|$,
    $J\gets J\cup \{j\}, S\gets S-\{s_j\}$;
    \EndWhile
    \end{algorithmic}
    \end{algorithm}
    \end{minipage}
}

Let $J_{j-1}=\{t(1), t(2), \ldots, t(j-1)\}$ and $U'_j$
denote $J$ and $U$ after the $(j-1)$-th iteration of 
this modified algorithm, respectively.
Let $re(x_{T(i_1, i_2, m)})$ denote the number of repetitions of
$x_{T(i_1, i_2, m)}$ in $U'_{j}$.
Since $J_{j-1}\cup (J^*-J_{j-1})$ is an optimal solution,
it is easy to see 
$
H(Att[i_1, J^*-J_{j-1}], Att[i_2, J^*-J_{j-1}])\geq re(x_{T(i_1, i_2, m)}).
$

In the $j$-th iteration, Algorithm~\ref{al:aps} will choose
$s_{t(j)}$ with the maximum coverage power,
which means that the average price will be minimized.
Then for $k'\in J^*-J_{j-1}$, we have
$
\alpha \leq  \frac{1}{|s_{k'}\cap U'_{j}|},
$
which indicates
$
\alpha \left( \sum_{k'\in J^*- J_{j-1}} {|s_{k'}\cap U'_{j}|} \right) \leq \sum_{k'\in J^*- J_{j-1}} 1 .
$
Thus we have
\begin{eqnarray}
\alpha & \leq & 
\frac{ |J^*-J_{j-1}|}{|( \bigcup_{k'\in J^*- J_{j-1}} s_{k'})\cap U'_{j}|}.
\label{eq:alpha}
\end{eqnarray}
Recall that
$
H(Att[i_1, J^*-J_{j-1}], Att[i_2, J^*-J_{j-1}])\geq re(x_{T(i_1, i_2, m)}).
$
This means that the number counting together all the repetitions of
$x_{T(i_1, i_2, m)}$ in set $s_k', k' \in J^*-J_{j-1}$ should be
greater than or equal to  $re(x_{T(i_1, i_2, m)})$, then we have
$
\left( \bigcup_{k'\in J^*- J_{j-1}} s_{k'} \right) \cap U'_{j}=U'_{j}.
$
Here the union operation should be over multi-sets.
From this and Ineq.~(\ref{eq:alpha}), we have
\begin{eqnarray*}
\alpha 
&\leq & \frac{|J^*-J_{j-1}|}{|( \bigcup_{k'\in J^*- J_{j-1}} s_{k'})\cap U'_{j}|}
~=~ \frac{|J^*-J_{i}|}{|U'_{j}|} \\
& \leq &\frac{|J^*|}{\displaystyle M(2K + 1)-\sum_{k=1}^{j-1}|s_{t(k)}\cap U'_k|}\\
&= & \frac{|J^*|}{\displaystyle M(2K + 1)-\sum_{k=1}^{j-1}|s_{t(k)}\cap U_k|}.
\end{eqnarray*}
Here $U_k$ is the set $U$ in the $k$-th iteration in Algorithm~\ref{al:aps}.
The last equality holds because there is no duplicated elements in $s_{t(k)}$,
and those elements are the same in $U_k^{'}$ and $U_k$ without considering
the number of duplications of each element.
Then the inequality is proved.

%
%

\end{document}